\documentclass[11pt,english]{article}
\usepackage[T1]{fontenc}
\usepackage[latin9]{inputenc}
\usepackage{geometry}
\usepackage{graphicx}
\usepackage{verbatim}
\usepackage{amsmath}
\usepackage{amssymb}

\makeatletter
\usepackage{times}
\usepackage{algorithm}
\usepackage{algorithmic}
\usepackage{amsthm}
\usepackage[all]{xy}

\makeatletter
 \theoremstyle{plain}
\newtheorem{thm}{Theorem}[section]
  \theoremstyle{plain}
  
 \theoremstyle{definition}
  
  \theoremstyle{definition}
  \newtheorem{defn}[thm]{Definition}
  \theoremstyle{definition}
  \newtheorem{rk}[thm]{Remark}
  \theoremstyle{definition}
  
  \theoremstyle{plain}
  
  \theoremstyle{plain}
  \newtheorem{lem}[thm]{Lemma}
    \theoremstyle{plain}
  
     \theoremstyle{plain}
  
     \theoremstyle{definition}

 \makeatother

\renewcommand{\epsilon}{\varepsilon}

\newcommand{\D}{\mathcal{D}}

\newcommand{\F}{\mathcal{F}}

\renewcommand{\H}{\mathcal{H}}

\newcommand{\M}{\mathcal{M}}

\renewcommand{\P}{\mathcal{P}}
\newcommand{\R}{\mathbb{R}}

\newcommand{\V}{\mathcal{V}}

\renewcommand{\div}{\text{div}}

\newcommand{\norm}[1]{{\left|\left|#1\right|\right|}}

\usepackage{babel}
\makeatother

\title{Some Results for the Primitive Equations with Physical Boundary Conditions}

\author{Lawrence Christopher Evans\footnote{Department of Mathematics, University of Missouri-Columbia. E-mail: evanslc@missouri.edu}, Robert Gastler\footnote{Department of Mathematics, University of Missouri-Columbia. E-mail: rrgm8d@mail.missouri.edu}}

\begin{document}

\maketitle

\abstract{In this paper we consider the (simplified) 3-dimensional primitive equations with \emph{physical boundary conditions}. We show that the equations with constant forcing have a bounded absorbing ball in the $H^1$-norm and that a solution to the unforced equations has its $H^1$-norm decay to $0$. From this, we argue that there exists an invariant measure (on $H^1$) for the equations under random kick-forcing.}

\section{Introduction and Statement of Results}

We consider the 3-dimensional primitive equations, a variant of the Navier Stokes equations in which the equation for the third component of velocity is removed and we make the assumption that the pressure $p$ is independent of the third space coordinate. In this paper we will, following the presentation of~\cite{KZ}, consider the following simplified version of the primitive equations:
\begin{equation}
\label{thePEsInitial}
\begin{cases}
\partial_t u_k-\nu\Delta u_k + \sum_{j=1}^3 u_j\partial_j u_k+ \partial_k p = \text{Forcing},\ \ k=1,2\\
 \div\ u= \partial_1 u_1+\partial_2 u_2+\partial_3 u_3 = 0\\
\end{cases}.
\end{equation}
We consider both the case of constant forcing, and forcing by random kicks.

In their breakthrough paper~\cite{CT}, Cao and Titi proved the existence of global strong solutions for the 3-dimensional primitive equations. Later, in~\cite{KZ}, Kukavica and Ziane proved the existence of global strong solutions under a different set of boundary conditions, which correspond more closely to physical models of the ocean and which we will refer to as the \emph{physical boundary conditions} (see (\ref{BCs}) below). The physical boundary conditions lead to different estimates\footnote{In particular, the equality $\int_M \nabla p_s\cdot\Delta \overline{v} dxdy=0$ in the beginning of Section 3.3.1 of~\cite{CT} does not hold for the physical boundary conditions.} and so Kukavica and Ziane give a proof that is substantially different from the one in~\cite{CT}.

Following the breakthrough of Cao and Titi, many results have been proved assuming their boundary conditions: In~\cite{Ju}, Ju considers bounded absorbing sets and global attractors. In~\cite{GH}, Gao and Huang consider the stochastically forced primitive equations and show the existence of random pullback attractors. More recently, in~\cite{DGTZ}, Debussche, Glatt-Holtz, Temam, and Ziane have proved the global existence of strong pathwise solutions to the primitive equations with forcing by multiplicative noise. Few papers, however, have considered the physical boundary conditions. 

In this paper we consider the 3-dimensional primitive equations with the physical boundary conditions. We show that under a constant forcing, the $V$-norm ($H^1$-norm) of the solution stays bounded and that under no forcing, the $V$-norm of the solution decays to $0$. These results are stated as Theorem \ref{goalTheorem} and Theorem \ref{absorbingforkicks}. Such results on the enstrophy have intuitive appeal. These results were proved in~\cite{Ju} for Cao and Titi's boundary conditions and the proof relies on applying the uniform Gronwall lemma to estimates\footnote{In particular, Ju uses the equality mentioned in the previous footnote (see inequality (3.11) of~\cite{Ju}).} from~\cite{CT}. The estimates in~\cite{KZ} seem less amenable to the uniform Gronwall lemma, and so we instead make a somewhat unorthodox argument that ties in closely to the argument in~\cite{KZ}.

Finally, we consider the primitive equations under a random kick forcing. We show that if the kicks are infrequent enough then there exists an invariant measure on $V$ ($H^1$ with extra conditions). Here we encounter the same set of issues discussed in~\cite{Ju}. We cannot define a dynamical system on $H$ ($L^2$ with extra conditions) due to the lack of uniqueness for weak solutions for the primitive equations. Therefore, we cannot even define an invariant measure on $H$. We can define a dynamical system and invariant measures on $V$ but then we have the issue that bounded balls in the $V$-norm are not compact in $V$. We get around this using a compactness argument from ~\cite{Ju}. Unfortunately, we do not see how to apply this argument in the case of forcing by white noise and so this leaves open the question of finding invariant measures for the primitive equations for such a forcing.

\subsection{Outline of our Paper}
In Section \ref{SectionSetup}, we introduce general notation and the various forms of the Primitive Equations we will consider, including the case of random kick-forcing. In Section \ref{Results}, we state precisely the results of this paper. In Section \ref{KZProof}, we reproduce a rough sketch of the argument in~\cite{KZ} in order to prove Lemma \ref{MainLemma}; Lemma \ref{MainLemma} is the key tool behind the proofs of our results and provides a quantitative bound on the growth of the $V$-norm of a solution over small time intervals. Finally, in Section \ref{ProofOfResults}, we prove our results.

\section{The Setup for the Primitive Equations}
\label{SectionSetup}
Mathematically, the Primitive Equations consist of taking the 3-dimensional Navier Stokes equations, removing the equation for the third component of the velocity, and positing that the pressure depends only on the first two position variables.

Let $G=G_2\times (-h,0)\subset \R^3$, where $h$ is a positive constant, and $G_2$ is a smooth bounded domain in $\R^2$. Let $u(x,t):=(v(x,t),u_3(x,t))$ be the velocity field where $v(x,t)$ is the horizontal velocity and $u_3(x,t)$ is the vertical velocity. We similarly decompose $x\in G$ as $x=(x',z)$. The primitive equations with constant forcing can then be written as
\begin{equation}
\label{thePEs}
\begin{cases}
\partial_t v-\nu\Delta v + (u\cdot \nabla) v+ \nabla_2 p = f\\
 \div\ u = 0\\
 v(x,0)=v_0
\end{cases},
\end{equation}
where $p=p(x',t)$ and where $v_0$ satisfies
$$
\div_2 \int_{-h}^0 v_0 dz = 0,
$$
and where $f\in H$ is independent of time.
The boundary conditions are described as follows. Let $\Gamma_t := G_2\times \{0\}$ denote the top, $\Gamma_s := \partial G_2 \times [-h,0]$ the sides, and $\Gamma_b:= G_2\times\{-h\}$ the bottom of $G$. The boundary conditions we consider, the so-called \emph{physical boundary conditions}, are
\begin{equation}
\label{BCs}
\begin{split}
\partial_z v =& 0 \text{ for } x\in\Gamma_t\\
v=&0 \text{ for } x\in\Gamma_b\cup \Gamma_s\\
u_3 =& 0 \text{ for } x\in \Gamma_t\cup \Gamma_b.
\end{split}
\end{equation}
Note that $u_3$ is determined by $v$ via the divergence free condition. Specifically we have that
$$
u_3(x,t)=-\int_{-h}^{z} \div_2 v (x',z',t) dz'
$$
and so we can express the primitive equations in the modified form, which can be thought of as a PDE for $v$,
\begin{equation}
\label{thePEsMod}
\begin{cases}
\partial_t v-\nu\Delta v + (v\cdot \nabla_2) v-\left(\int_{-h}^{z} \div_2 v (x',z',t) dz'\right)\partial_z v+ \nabla_2 p = f\\
 v(x,0)=v_0
\end{cases}.
\end{equation}

\subsection{The Functional Analytic Setup}
We now introduce some spaces of interest. Let
$$
\V := \left\{v\in C^\infty_{b,s}(G): \div_2\int_{-h}^0 v\ dz = 0 \text{ on } G_2\right\},
$$
where
$$
C^\infty_{b,s}(G):=\{v\in C^\infty(G): \text{supp}(v)\text{ is compact in }\bar{G}-(\Gamma_b\cup\Gamma_s)\}.
$$
We then define the spaces $H:=\overline{\V}^{L^2}$ and $V:=\overline{\V}^{H^1}$ (i.e. the closures in those topologies). It has been shown that (see~\cite{LTW} or Lemma 2.1 in ~\cite{HTZ2})
$$
H=\left\{v\in L^2: \div_2 \int_{-h}^0 v\ dz = 0 \text{ on }G_2, \left(\int_{-h}^0 v\ dz\right)\cdot n = 0 \text{ on }\partial G_2\right\}
$$
and
$$
V=\left\{v\in H\cap H^1 : v=0\text{ on } \Gamma_b\cup\Gamma_s \right\}.
$$
Also, $L^2 = H\oplus H^{\perp}$, where 
$$
H^{\perp} = \{v\in L^2: v=\nabla_2 p\text{ with } p\in H^1(G_2)\}.
$$
We equip these spaces with the norms
$$
\norm{v}_H := \norm{v}_{L^2}
$$
and
$$
\norm{v}_V := \left(\int_G |\nabla v|^2 dx\right)^{\frac{1}{2}},
$$
which is equivalent to the $H^1$-norm by the Poincar\'e inequality. Finally we introduce the spaces $V^n:=V\cap H^n$ with the usual Sobolev norm.

Let $A := -\Pi_H \Delta$ be the negative projection of the Laplacian onto the space $H$. Let (noting that here $u$ and $v$ are just placeholder variables)
$$
B(u,v):=\Pi_H\left[(u\cdot \nabla_2)v - \left(\int_{-h}^z \div_2 u\ dz'\right)\partial_z v\right].
$$

We can now project (\ref{thePEsMod}) onto $H$ to get
\begin{equation}
\label{thePEsAbstract}
\begin{cases}
\partial_t v+\nu Av + B(v,v) = \Pi_H f\\
 v(x,0)=v_0
\end{cases},
\end{equation}
which can be analyzed as an abstract evolution equation.

\subsection{Types of Solution}

\begin{defn}
We say that $v$ is a \emph{weak solution} to (\ref{thePEsMod}) on $[0,T]$ if
$$
v\in L^\infty([0,T];H)\cap L^2([0,T];V),\ \partial_t v\in L^2([0,T];V^{-3})
$$ 
(here $V^{-3}$ denotes the dual space to $V^3$) and the equalities in $(\ref{thePEsMod})$ hold in $V^{-3}$. That is, $\forall w\in V^3$,
\begin{align*}
\langle\partial_t v+ (v\cdot \nabla_2) v-\left(\int_{-h}^{z} \div_2 v (x',z',t) dz'\right)\partial_z v,w\rangle+ \langle\nu \nabla v,\nabla w\rangle =& \langle f,w\rangle\ \text{a.e. }t\in[0,T] \\
 \langle v(x,0), w\rangle = \langle v_0, w \rangle.
\end{align*}
Note that $\nabla p = 0$ in $V^{-3}$ (so it drops out) and note that our regularity assumptions on $v$ give us that $v\in C([0,T];V^{-3})$ so it makes sense to talk about $v(0)$. 
\end{defn}

\begin{defn}
We say that $v$ is a \emph{strong solution} to (\ref{thePEsMod}) on $[0,T]$ if 
$$
v\in L^\infty([0,T];V)\cap L^2([0,T];\D(A)),\ \partial_t v\in L^2([0,T];H)
$$ 
and the equalities in $(\ref{thePEsMod})$ hold in $H$. That is, 
\begin{align*}
\partial_t v+\nu\Delta v+ (v\cdot \nabla_2) v-\left(\int_{-h}^{z} \div_2 v (x',z',t) dz'\right)\partial_z v+ \nabla_2 p=& f\ \text{a.e. }x\in G, t\in[0,T] \\
v(x,0) = v_0\ \text{a.e. }x\in G.
\end{align*}
\end{defn}

It has been shown in~\cite{KZ} that for any $v_0\in V$ there exists a unique global\footnote{Here by ``global'' we mean there exists a solution on $[0,T]$ for all $T>0$} strong solution. In fact it can also be shown that this solution lies in $C(0,T;V)$ (This is crucial as we will routinely talk about the value $v(t)$ at different times $t$. See Appendix \ref{contintimeSection} for a proof). Global existence of weak solutions for any $v_0\in H$ was proven earlier (see~\cite{BGMR},~\cite{HTZ}, or~\cite{TZ}) but uniqueness of weak solutions remains an open problem. Therefore we consider only strong solutions and work in the space $V$ in order that we have a well defined dynamical system.

\subsection{Kick Forcing}

We now consider the primitive equations under a random kick-forcing. For a more in depth explanation of random kick-forcing we refer the reader to Chapter 3 of the book by Kuksin,~\cite{KuksinBook}. Let $\{\xi_n\}_{n=1}^\infty$ be i.i.d. random variables on a fixed probability space $(\Omega,\F,\mathbb{P})$ which take values in $V$. Let $S(t):V\to V$ be the solution operator to the primitive equations (\ref{thePEs}) with no forcing (i.e. $f\equiv 0$). This solution operator is well defined by global existence and uniqueness for strong solutions and therefore defines a dynamical system.

Fix a time step $T>0$. We define a random dynamical system, corresponding to random kicks at time intervals $T$: Fix a $v_0\in V$ and let $X_n:\Omega\to V$ be the random variables given by the relations
\begin{equation}
\label{eqnforXn}
X_0\equiv v_0\text{ and }X_{n}(\omega)=S(T)\left[X_{n-1}(\omega)\right]+\xi_n(\omega),\ \text{for }n=1,2,\ldots
\end{equation}
(We will suppress the dependence on $T$ of various objects. As $T$ is fixed at the outset this should provide no confusion). That is, at every time step $T$, we give our dynamical system a kick but otherwise flow according to the solution operator. $X_n$ is then a time-independent discrete time Markov process indexed by the positive integers and it has the associated transition probabilities $P(v,A)=\mathbb{P}(S(T)[v]+\xi_1\in A)$. Associated to this Markov process is the operator
\begin{align*}
\P:L^\infty(V)&\to L^\infty(V)\\
[\P f](v)&:=\int_V f(v')P(v,dv')
\end{align*}
and, letting $\M(V)$ denote the space of probability measures on $V$, we have also the dual operator
\begin{align*}
\P^*:\M(V)&\to \M(V)\\
[\P^*\mu](A):= \langle \P 1_A,\mu\rangle &= \int_V P(v,A)\mu(dv).
\end{align*}
Because of continuity of the mapping $v\mapsto S(T)[v]$ (see Appendix \ref{A2}) we have, by the Lebesgue dominated convergence theorem, that $\P:C_b(V)\to C_b(V)$, i.e. our Markov process is Feller.

Finally, we say that $\mu\in \M(V)$ is an \emph{invariant measure} for the kick-forced primitive equations with kicks at time intervals $T$ if $\P^*\mu=\mu$. 

\section {Our Results}
\label{Results}

We now state our main theorems: 

\begin{thm} 
\label{goalTheorem}
(Bounded absorbing set in $V$ under constant forcing) Let $v_0\in V$ be such that $\norm{v_0}^2_V\leq R$, let $f\in H$, and let $v(t)$ be the solution to (\ref{thePEsMod}) with initial data $v_0$. Then there exists $K>0$ and $T_V>0$, depending only on $R$ and $\norm{f}_H$, such that 
$$
\norm{v(t)}_V^2 \leq K,\ \ \forall t>T_V.
$$
\end{thm}

\begin{thm}
\label{absorbingforkicks}
(Decay in $V$-norm for the unforced equation) Let $v_0\in V$ be such that $\norm{v_0}^2_V\leq R$, let $f\equiv 0$, and let $v(t)$ be the solution to (\ref{thePEsMod}) with initial data $v_0$. Then for all $\epsilon>0$, there exists a time\footnote{We abuse notation here by defining $T_V$ again but we will never be considering the forced and unforced equations simultaneously so there should be no confusion.} $T_V=T_V(R,\epsilon)$ such that
$$
\norm{v(t)}_V^2 \leq \epsilon,\text{ for all } t\geq T_V(R,\epsilon).
$$
\end{thm}
\begin{rk}
Note carefully that this time $T_V$ depends only on $R$ and not the actual value of $u_0$ itself. This stronger statement will be key for us. If we only wanted $T_V$ to depend on $u_0$ we could argue this more directly from the paper~\cite{KZ} (the issue is that the $\delta$ they define on page 2746 depends on $u_0$ itself).
\end{rk}
\begin{thm}
\label{InvariantMeasureKickForcing}
(Invariant Measure for Kick-forcing) Consider the primitive equations with random kick-forcing and assume the kicks are bounded in the $H^2$ norm, i.e.
$$
\text{There exists }R>0,\text{ such that } \norm{A\xi_n}^2_H\leq R \text{ a.s.},\ \forall n.
$$
Then there exists a time $T=T(R)$ such there exists an invariant measure for the primitive equations with this random kick-forcing at time intervals $T$.
\end{thm}

\section{Sketch of the Proof of Global Existence of Strong Solutions}
\label{KZProof}

Our main goal in this section is to state and prove Lemma \ref{MainLemma}. Since Lemma \ref{MainLemma} hinges on an inequality which appears deep in~\cite{KZ} (inequality (2.14) there, which we give as inequality (\ref{IneqJandK}) here), we have decided things would be clearest if we reproduced the sketch of the proof of global existence from~\cite{KZ} here. As a small bonus, we will show how Lemma \ref{MainLemma} can be used to give an alternate ending to the proof. Since the argument until inequality (\ref{IneqJandK}) is more or less verbatim, we present only a rough sketch of the argument and direct the reader to~\cite{KZ} for more details.

Let $v(t)$ be a local strong solution to (\ref{thePEsMod}) with initial condition $v_0 \in V$, extended to its maximal interval of existence $[0,T_{\text{max}})$. We will prove that $\norm{v(t)}_V$ is bounded on this interval yielding a contradiction.

We take the inner product of both sides of (\ref{thePEsAbstract}) by $Av$ and use standard estimates to get\footnote{We will use $C$ to denote a positive constant which may change from line to line but does not depend on any critical quantities.}
\begin{equation}
\label{Gronwall}
\frac{d}{dt}\norm{v}_V^2+\nu\norm{Av}_H^2 \leq \frac{C}{\nu^3}\left(\norm{v}_{L^6}^4+\norm{\partial_z v}_H^2\norm{\nabla\partial_z v}_{H}^2\right)\norm{v}_V^2+\norm{f}_H^2.
\end{equation}
In order to apply the Gronwall lemma, we now estimate the terms in parentheses.

\subsection{The $\norm{v}_{L^6}$ Estimate}

We start with the primitive equations in the form (\ref{thePEs}), multiply both sides by $u_k^5$ for $k=1,2$, integrate over $G$, and sum over $k$ to get that
\begin{equation}
\label{eqn1}
\frac{1}{6}\sum_{k=1}^2\frac{d}{dt}\norm{u_k}_{L^6}^6+\frac{5}{9}\nu\sum_{k=1}^2\int_G |\nabla(u_k^3)|^2\ dx\\
                                          = -\sum_{k=1}^2\int_G\partial_k p u_k^5\ dx+ \sum_{k=1}^2\int_G f_k u_k^5\ dx := I_1+ I_2.
\end{equation}
The second term on the right side of equation (\ref{eqn1}) is estimated by 
$$
I_2\leq C \norm{f}_H\norm{v}_{L^6}^2\left(\sum_{k=1}^2\norm{\nabla(u_k^3)}_{L^2}^2\right)^{\frac{1}{2}},
$$
while the first term on the right side of (\ref{eqn1}) is handled by an averaging trick (This trick was pioneered in~\cite{CT} and it is here we exploit the fact that $p$ is independent of $z$): We let
$$
M(u)(x'):=\frac{1}{h}\int_{-h}^{0} u(x',z)\ dz,
$$
and so, since the pressure is independent of $z$, 
$$
I_1=-h\sum_{k=1}^2\int_{G_2} M(u_k^5)\partial_k p\ dx' \leq h\sum_{k=1}^2 \norm{M(u_k^5)}_{L^3(G_2)}\norm{\nabla_2 p}_{L^{\frac{3}{2}}(G_2)}.
$$
Now, playing with Sobolev inequalities (taking advantage of the fact we are considering a 2 dimensional domain) we get that
$$
I_1\leq \frac{C}{\nu}\norm{\nabla_2 p}_{L^{\frac{3}{2}}(G_2)}^2\norm{v}_{L^6}^4+\epsilon\nu\sum_{k=1}^2\norm{\nabla(u_k^3)}_{L^2}^2.
$$
Thus we arrive at the inequality
\begin{equation}
\label{estimateJ}
\frac{1}{6}\sum_{k=1}^2\frac{d}{dt}\norm{u_k}_{L^6}^6+\frac{5}{9}\nu\sum_{k=1}^2\int_G |\nabla(u_k^3)|^2\ dx                                          \leq \frac{C}{\nu}\left(\norm{\nabla_2 p}_{L^{\frac{3}{2}}(G_2)}^2+\norm{f}_H^2\right)\norm{v}_{L^6}^4.
\end{equation}

\subsection{The $\norm{\partial_z v}_{L^2}$ and $\norm{\nabla\partial_z v}_{L^2}$ Estimates}

We multiply equation (\ref{thePEs}) by $-\partial_{zz}u_k$, $k=1,2$, integrate over $G$, and sum to get
\begin{multline}
\label{eqn5}
-\sum_{k=1}^2\int_G \partial_t u_k \partial_{zz}u_k\ dx+\nu\sum_{k=1}^2\int_G \Delta u_k \partial_{zz}u_k\ dx\\
                                          = \sum_{k=1}^2\sum_{j=1}^3\int_G u_j\partial_j u_k \partial_{zz}u_k\ dx +\sum_{k=1}^2\int_G\partial_k p \partial_{zz}u_k\ dx- \sum_{k=1}^2\int_G f_k \partial_{zz}u_k\ dx:= I_1+I_2+I_3.
\end{multline}
We estimate $I_2$: Recalling that the pressure is independent of $z$ we have that
\begin{align*}
\int_G\partial_k p \partial_{zz}u_k\ dx \leq& \left|\int_{G_2}\partial_k p \int_{-h}^0\partial_{zz}u_k\ dz\ dx'\right| \\
          =& \left|\int_{G_2}\partial_k p \partial_z u_k(z=-h)\ dx'\right| \\
          \leq& \norm{\partial_k p}_{L^{\frac{3}{2}}(G_2)}\norm{\partial_z u_k (z=-h)}_{L^3(G_2)}.
\end{align*}
whence, by two-dimensional Sobolev embeddings,
\begin{align*}
I_2 \leq& \frac{C}{\nu^3}\norm{\nabla_2 p}^2_{L^{\frac{3}{2}}(G_2)}+\epsilon\nu\norm{\nabla\partial_z v}_{L^2}^2.
\end{align*}
After using integration by parts and standard estimates to estimate $I_1$ and $I_3$, we have
\begin{equation}
\label{estimateK}
\frac{d}{dt}\norm{\partial_z v}^2_{L^2}+\nu\norm{\nabla \partial_z v}_{L^2}^2
  \leq \frac{C}{\nu^3}\norm{v}_{L^6}^4\norm{\partial_z v}_{L^2}^2 + \frac{C}{\nu}\norm{\nabla_2 p}_{L^{\frac{3}{2}}(G_2)}^2+C\norm{f}_{L^2}^2.
\end{equation}

\subsection{The $\norm{\nabla_2 p}_{L^{\frac{3}{2}}(G_2)}$ Estimate}

We will need to estimate the pressure as it shows up in the previous two estimates. We apply the vertical averaging operator to both sides of equation (\ref{thePEs}) to get, after integration by parts on the $j=3$ summand,
\begin{equation}
\label{thePEsAverage}
\begin{cases}
M\partial_t u_k -\nu M \Delta_2 u_k + \partial_k p =& \nu\partial_z u_k(z=-h) - M\left(\sum_{j=1}^2 u_j\partial_j u_k\right)\\
&- M\left(\sum_{j=1}^2 \partial_j u_j u_k\right)+ Mf_k,\text{ for }k=1,2\\
\partial_1 M u_1 + \partial_2 M u_2 = 0.
\end{cases}
\end{equation}
We note that, for any time interval $[\tau_1,\tau_2]$, the PDE (\ref{thePEsAverage}) can be thought of as a 2D Stokes problem on $G_2\times [\tau_1,\tau_2]$ for $(Mu,p)$ with initial data $Mu(\tau_1)$. Therefore, we can apply a regularity result of Sohr and von Wahl from~\cite{SvW} which says that
$$
\int_{\tau_1}^{\tau_2}\norm{\nabla_2 p}^2_{L^{\frac{3}{2}}(G_2)}\ d\tau \leq C\norm{-\Delta_{3/2}^{1/2+\epsilon}Mu(\tau_1)}^2_{L^{\frac{3}{2}}(G_2)}+C\int_{\tau_1}^{\tau_2}\norm{RHS}^2_{L^{\frac{3}{2}}(G_2)}\ d\tau,
$$
where $RHS$ is the right hand side of (\ref{thePEsAverage}). After estimating these terms, we get the pressure estimate
\begin{equation}
\label{estimateP}
\begin{split}
\int_{\tau_1}^{\tau_2}\norm{\nabla_2 p}^2_{L^{3/2}(G_2)}\ d\tau \leq& C\norm{v(\tau_1)}^2_V + C\int_{\tau_1}^{\tau_2}\norm{\partial_z v}_{L^2}\norm{\nabla\partial_z v}_{L^2}\ d\tau\\
&+C\int_{\tau_1}^{\tau_2}\norm{v}^2_{L^6}\norm{v}^2_{V}\ d\tau + C\int_{\tau_1}^{\tau_2}\norm{f}^2_{L^2}\ d\tau.
\end{split}
\end{equation}

\subsection{Combining the three estimates to bound the quantities $\norm{v}^4_{L^6}$, $\norm{\partial_z v}^2_{L^2}$, and $\norm{\nabla\partial_z v}^2_{L^2}$}

We consider our three estimates (\ref{estimateJ}), (\ref{estimateK}), and (\ref{estimateP}). To simplify notation we use the notation from~\cite{KZ}:
\begin{eqnarray*}
J:= \norm{v}_{L^6}\\
K:=\norm{\partial_z v}_{L^2}\\
\bar{K} := \norm{\nabla\partial_z v}_{L^2}\\
\bar{E} := \norm{v}_V
\end{eqnarray*}
and we also use the notation
$$
\norm{G}^2_{L^2_t(\tau_1,\tau_2)}:=\int_{\tau_1}^{\tau_2} G^2(\tau)\ d\tau.
$$
Consider an arbitrary triple of times $0\leq \tau_1<\tau_2<\tau_3$. From (\ref{estimateJ}) we have
$$
J^4(\tau_2)\leq J^4(\tau_1) + \frac{C}{\nu}\left(C\bar{E}^2(\tau_1)+ C\nu\norm{K}_{L^2_t}\norm{\bar{K}}_{L^2_t}+C\norm{J\bar{E}}^2_{L^2_t}\right)\sup_{\tau_1\leq\tau\leq\tau_3}J^2(\tau)\\
$$
and after taking the supremum over $\tau_2\in[\tau_1,\tau_3]$, rearranging terms, and applying standard estimates (and in particular using $K\leq C\bar{E}$) we get
\begin{equation}
\label{IneqJ}
\sup_{\tau_1\leq\tau\leq\tau_3}J^4(\tau)\leq J^4(\tau_1)+ \left(C\norm{\bar{E}}^2_{L^2_t}+\epsilon\right)\sup_{\tau_1\leq\tau\leq\tau_3}J^4(\tau)+\frac{C}{\nu^2}\bar{E}^4(\tau_1)+ C\norm{\bar{E}}^2_{L^2_t}\norm{\bar{K}}^2_{L^2_t}+\frac{C}{\nu^2}\norm{f}^4_{L^2_t}.
\end{equation}
Similarly, starting from (\ref{estimateK}), rearranging terms and applying standard estimates we get
\begin{multline}
\label{IneqK}
\sup_{\tau_1\leq\tau\leq\tau_3}K^2(\tau)+\frac{\nu}{2}\norm{\bar{K}}^2_{L^2_t}\leq K^2(\tau_1)+\left(\frac{C}{\nu^3}\norm{\bar{E}}^2_{L^2_t}+\epsilon\right)\sup_{\tau_1\leq\tau\leq\tau_3}J^4(t)\\
+ \frac{C}{\nu}\bar{E}(\tau_1)^2 + \frac{C}{\nu}\norm{\bar{E}}^2_{L^2_t}+ \frac{C}{\nu^2}\norm{\bar{E}}^4_{L^2_t}+\frac{C}{\nu}\norm{f}^2_{L^2_t}.
\end{multline}
Summing these last two inequalities we get
\begin{equation}
\label{IneqJandK}
\begin{split}
\sup_{\tau_1\leq\tau\leq\tau_3}J^4(\tau)+\sup_{\tau_1\leq\tau\leq\tau_3}K^2(\tau)+\frac{\nu}{2}\norm{\bar{K}}^2_{L^2_t}\leq&\ J^4(\tau_1)+K^2(\tau_1)\\
&+\left(\frac{C}{\nu}\norm{\bar{E}}^2_{L^2_t}+\frac{C}{\nu^3}\norm{\bar{E}}^2_{L^2_t}+2\epsilon\right)\sup_{\tau_1\leq\tau\leq\tau_3}J^4(t)\\
&+C\norm{\bar{E}}^2_{L^2_t}\norm{\bar{K}}^2_{L^2_t}+ \frac{C}{\nu}\bar{E}(\tau_1)^2 + \frac{C}{\nu}\norm{\bar{E}}^2_{L^2_t}\\
&+\frac{C}{\nu^2}\bar{E}^4(\tau_1)+\frac{C}{\nu^2}\norm{\bar{E}}^4_{L^2_t}+\frac{C}{\nu}\norm{f}^2_{L^2_t}+\frac{C}{\nu^2}\norm{f}^4_{L^2_t}.
\end{split}
\end{equation}
We now wish to consider intervals $[\tau_1,\tau_3]$ small enough that we can keep the right side of (\ref{IneqJandK}) under control. Specifically, we want the three terms in the parentheses to be less than $\frac{1}{2}$ and $C\norm{\bar{E}}^2_{L^2_t}\leq\frac{\nu}{4}$ so that we can absorb the third and fourth term on the right into the left hand side. 

\subsection{The Growth Control Lemma and the Conclusion of the Proof}
\label{ConclusionOfProof}

All of the above estimates appear in~\cite{KZ}. From here on we present a slightly different approach than that of Kukavica and Ziane. We put the ideas of the previous paragraph into the following lemma:
\begin{lem}
\label{MainLemma} (Growth control lemma)
There is some $\eta>0$ such that if $0\leq\tau_1\leq \tau_3$ are close in the sense that
\begin{equation}
\label{etacondition}
|\tau_3-\tau_1|\leq 1\text{ and }\int_{\tau_1}^{\tau_3} \norm{v(\tau)}^2_V\ d\tau \leq \eta,
\end{equation}
then
$$
\norm{v(\tau_2)}^2_V \leq e^{C(1+\norm{v(\tau_1)}_V^2)^4}\left[\norm{v(\tau_1)}_V^2+\norm{f}_H^2\right]=:\Gamma\left(\norm{v(\tau_1)}_V^2\right)
$$
for any $\tau_2\in[\tau_1,\tau_3]$, where $C=C(\nu,\eta,\norm{f}_H)$.
\end{lem}

\begin{rk}
Note what this lemma says: Provided $\tau_1$ and $\tau_3$ are ``close enough'', the $V$-norm of $v$ only blows up so much in going from $\tau_1$ to $\tau_3$. In particular, if $\norm{v(\tau_1)}_V$ is finite then so too is $\sup_{\tau_1\leq\tau_2\leq\tau_3}\norm{v(\tau_2)}_V$. We will use this idea extensively.

Also note that the mapping $y\mapsto \Gamma(y)$ is non-decreasing and when $f\equiv 0$, $\lim_{y\to 0}\Gamma(y)=0$.
\end{rk}

\begin{proof}
First note that there exists an $\eta=\eta(C,\nu)$ small enough such that when $\tau_1,\tau_3$ satisfy (\ref{etacondition}), the conditions stated in the paragraph preceding the theorem are satisfied. Then, after absorbing terms into the left hand side we have
\begin{equation*}
\begin{split}
\frac{1}{2}\sup_{\tau_1\leq\tau\leq\tau_3}J^4(\tau)+\sup_{\tau_1\leq\tau\leq\tau_3}K^2(\tau)+\frac{\nu}{4}\norm{\bar{K}}^2_{L^2_t}\leq&\ J^4(\tau_1)+K^2(\tau_1)\\
&+ \frac{C}{\nu}\bar{E}(\tau_1)^2 + \frac{C}{\nu}\norm{\bar{E}}^2_{L^2_t}\\
&+\frac{C}{\nu^2}\bar{E}^4(\tau_1)+\frac{C}{\nu^2}\norm{\bar{E}}^4_{L^2_t}+\frac{C}{\nu}\norm{f}^2_{L^2_t}+\frac{C}{\nu^2}\norm{f}^4_{L^2_t}.
\end{split}
\end{equation*}
Next we note that $K(\tau_1)\leq \bar{E}(\tau_1)$ trivially and by a Sobolev inequality $J(\tau_1)\leq \bar{E}(\tau_1)$. From our assumption we have that $\norm{\bar{E}}^2_{L^2_t}\leq \eta$. And so, absorbing $\nu$, $\eta$, $\norm{f}_H$, and $C$ into a new constant $C$, we can rewrite our inequality as
\begin{equation*}
\frac{1}{2}\sup_{\tau_1\leq\tau\leq\tau_3}J^4(\tau)+\sup_{\tau_1\leq\tau\leq\tau_3}K^2(\tau)+\frac{\nu}{4}\norm{\bar{K}}^2_{L^2_t}\leq C\bar{E}(\tau_1)^2+ C\bar{E}^4(\tau_1)+ C=:(\star).
\end{equation*}
Applying Gronwall's lemma to (\ref{Gronwall}) and using this estimate we get
\begin{align*}
\norm{v(\tau_2)}_V^2\leq& e^{C\int_{\tau_1}^{\tau_3}\left(\norm{v}_{L^6}^4+\norm{\partial_z v}_H^2\norm{\nabla\partial_z v}_{H}^2\right)\ d\tau}\left[\norm{v(\tau_1)}_V^2+(\tau_3-\tau_1)\norm{f}_H^2\right]\\
\leq& e^{\left(C(\tau_3-\tau_1)(\star)+ C(\star)(\star)\right)}\left[\norm{v(\tau_1)}_V^2+(\tau_3-\tau_1)\norm{f}_H^2\right],
\end{align*}
from which the theorem follows.
\end{proof}
 
With Lemma \ref{MainLemma} in hand, the proof of global existence is simple: By the a priori estimate
\begin{equation}
\label{apriorienergy}
\int_0^{T_{\text{max}}}\bar{E}^2(\tau)\ d\tau < \infty
\end{equation}
which follows immediately from multiplying both sides of the original primitive equations by $v$ (see (\ref{multbyv}) for further details), we can partition the interval $[0,T_{\text{max}})$ into a finite number, say $L$, intervals of the form $[t_\ell,t_{\ell+1})$ where $t_\ell$ and $t_{\ell+1}$ satisfy the conditions of Lemma \ref{MainLemma} and $t_L=T_{\text{max}}$. Since $\norm{v(0)}_V<\infty$, after iterating Lemma \ref{MainLemma} $L$ times we get that
\begin{equation}
\label{boundsupnorm}
\sup_{t\in[0,T_{\text{max}})}\norm{v(t)}_V^2\leq \Gamma^{(L)}\left(\norm{v(0)}_V^2\right),
\end{equation}
where $\Gamma^{(L)}(\cdot)$ denotes the $L$-fold composition of the function $\Gamma(\cdot)$ defined in Lemma \ref{MainLemma}. This implies that $\norm{v(t)}_V$ does not blow up as $t$ approaches $T_{\text{max}}$. This completes the proof of the global existence of strong solutions. Note that while our estimates here are not as sharp as those in Kukavica and Ziane, Lemma \ref{MainLemma} will be instrumental in proving our main result and so we thought it helpful to present an application here.
 
\section{Proof of our Results}
\label{ProofOfResults} 
 
\subsection{Proof of Theorem \ref{goalTheorem}}

We first prove Theorem \ref{goalTheorem}: The assumption $\norm{v_0}^2_V\leq R<\infty$ implies $\norm{v_0}^2_H\leq R$ as well.

First we show that we have an absorbing ball in $H$. If we take the inner product of both sides of (\ref{thePEsAbstract}) by $v$ we get
$$
\frac{1}{2}\partial_t\norm{v}_H^2+\norm{v}_V^2\leq (\Pi_H f,v)_H\leq \frac{1}{2} \norm{v}_H^2 + \frac{1}{2}\norm{f}_H^2,
$$
whence
\begin{equation}
\label{multbyv}
\frac{1}{2}\partial_t\norm{v}_H^2+\frac{1}{2}\norm{v}_V^2\leq \frac{1}{2}\norm{f}_H^2,
\end{equation}
and so, by the Poincar\'e inequality,
$$
\partial_t\norm{v}_H^2\leq -\lambda_1\norm{v}_H^2+\norm{f}_H^2
$$
where $\lambda_1$ is the first eigenvalue of $A$. From this it follows from basic ODE theory that there exists a time $T_H=T_H(R,\norm{f}_H)$ such that
$$
\norm{v(t)}_H^2\leq K\text{ for all } t\geq T_H.
$$
Also from integrating both sides of (\ref{multbyv}) we have that
\begin{equation}
\label{piece1}
\int_s^t \norm{v(\tau)}_V^2\ d\tau \leq \norm{v(s)}_H^2+(t-s)\norm{f}_H^2.
\end{equation}
If we take $T_H\leq s\leq t\leq s+1$ the right hand side is bounded by $K+\norm{f}_H^2$. However this only gives us control of the integral of $\norm{v(t)}_V^2$ and not control pointwise. Nevertheless, we will get pointwise control after combining (\ref{piece1}) with Lemma \ref{MainLemma}.

Consider the series of times $T_H\leq T-2 < T$ where $T$ is otherwise arbitrary. From (\ref{piece1}) we have that
$$
\int_{T-2}^{T-1} \norm{v(\tau)}_V^2\ d\tau \leq K+\norm{f}_H^2.
$$
It follows that there exists a time $t_0\in[T-2,T-1]$ such that $\norm{v(t_0)}_V^2 \leq K+\norm{f}_H^2$. Considering (\ref{piece1}) again we see that
$$
\int_{t_0}^T \norm{v(\tau)}_V^2\ d\tau \leq K+2\norm{f}_H^2<\infty.
$$
Therefore, we can divide the interval $[t_0,T]$ into a finite number, say $L$, intervals of the form $[t_\ell,t_{\ell+1}]$ whose endpoints satisfy the conditions of Theorem \ref{MainLemma}. Note that the number of intervals $L$ needed depends only on the quantity $K+2\norm{f}_H^2$ and therefore is independent of $T$. It follows that we can get from  $t_0$ to $T$ in $L$ or fewer ``steps'' and so from Theorem \ref{MainLemma} we have that
$$
\norm{v(T)}_V^2\leq \Gamma^{(L)}\left(\norm{v(t_0)}_V^2\right)\leq \Gamma^{(L)}\left(K+\norm{f}_H^2\right).
$$
Since $T$ was arbitrary other than needing to be larger than $T_H$ by $2$, Theorem \ref{goalTheorem} is proven with $T_V:=T_H+2$ and $K$ replaced by $\Gamma^{(L)}\left(K+\norm{f}_H^2\right)$.

\subsection{Proof of Theorem \ref{absorbingforkicks}}

We now prove Theorem \ref{absorbingforkicks}: First we prove the same theorem but for $H$-norms. Arguing as in the previous proof but with $f\equiv 0$, we get the inequalities
\begin{equation}
\label{piece2}
\int_s^t \norm{v(\tau)}_V^2\ d\tau \leq \norm{v(s)}_H^2
\end{equation}
and
$$
\partial_t\norm{v}_H^2\leq -\lambda_1\norm{v}_H^2.
$$
It follows from basic ODE theory that for each $\epsilon>0$ and $v_0$ with $\norm{v_0}_H^2\leq R$, there exists a time $T_H=T_H(R,\epsilon)$ such that
$$
\norm{v(t)}_H^2\leq \epsilon\text{ for all } t\geq T_H
$$
as desired.

Next we consider the series of times $T_H\leq T-2 < T$ where $T$ is otherwise arbitrary. From (\ref{piece2}) we have that
$$
\int_{T-2}^{T-1} \norm{v(\tau)}_V^2\ d\tau \leq \epsilon.
$$
It follows that there exists a time $t_0\in[T-2,T-1]$ such that $\norm{v(t_0)}_V^2 \leq \epsilon$. Considering (\ref{piece2}) again we see that
$$
\int_{t_0}^{T} \norm{v(\tau)}_V^2\ d\tau \leq \epsilon<\infty.
$$
Therefore, we can divide the interval $[t_0,T]$ into a finite number, $L$, intervals of the form $[t_\ell,t_{\ell+1}]$ whose endpoints satisfy the conditions of Lemma \ref{MainLemma} (Note that any integer $L\geq\max(\frac{\epsilon}{\eta},2)$ will work and so we can take $L$ independent of $T$). It follows then that we can get from  $t_0$ to $T$ in $L$ or fewer ``steps'' and so from Lemma \ref{MainLemma} we have that
$$
\norm{v(T)}_V^2\leq \Gamma^{(L)}\left(\norm{v(t_0)}_V^2\right)\leq \Gamma^{(L)}\left(\epsilon\right).
$$
Since $T$ was arbitrary other than needing to be larger than $T_H$ by $2$, Theorem \ref{absorbingforkicks} now holds with $T_V\left(R,\Gamma^{(L)}\left(\epsilon\right)\right)=T_H(R,\epsilon)+2$. This implies the given statement of the theorem since $\Gamma^{(L)}\left(\epsilon\right)\to 0$ as $\epsilon\to 0$.

\subsection{Proof of Theorem \ref{InvariantMeasureKickForcing}}

First note that $S(t):V\to V$ is a compact operator for all $t>0$ (see Appendix \ref{A3}). Fix $v_0\in V$ such that $\norm{v_0}_V^2\leq R$, take $T=T_V(4R,R)$, and define the Markov process $X_n(\omega)$ by the relations (\ref{eqnforXn}).
Then it follows that if $\norm{X_n(\omega)}_V^2\leq 4R$ then
$$
\norm{X_{n+1}(\omega)}_V^2=\norm{S(T)[X_n(\omega)]+\eta_{n+1}}_V^2\leq2\norm{S(T)[X_n(\omega)]}_V^2+2\norm{\eta_{n+1}}_V^2\leq 2R+2R=4R
$$
as well. And so it follows by induction that $\norm{X_n(\omega)}_V^2\leq 4R$ for all $n$. 

Let $\mu_n\in\M(V)$ be the distribution of $X_n$ and note that $\mu_{n+1}=\P^*\mu_n$. It follows from the above argument that each $\mu_n$ is supported on the compact set $S(T)\left[B_V(2\sqrt{R})\right]+B_{\D(A)}(\sqrt{R})$ where $B_V(\rho):=\{v\in V: \norm{v}_V\leq \rho\}$ and $B_{\D(A)}(\rho):=\{v\in V: \norm{Av}_H\leq \rho\}$. Therefore, the sequence of measures $\{\mu_n\}\subset \M(V)$ is tight.

We now use the standard method of Krylov and Bugolybov to show the existence of an invariant measure, taking advantage of the tightness of the $\mu_n$ and the Feller property of our Markov process. By tightness, there exists a subsequence of the $\mu_n$ (which we relabel to be the sequence of natural numbers) and there exists a measure $\mu\in \M(V)$, such that $\mu_n\rightharpoonup \mu$. Next we define the measures $\bar{\mu}_n = \frac{1}{n}\sum_{k=1}^n \mu_k$. We have then also that $\bar{\mu}_n\rightharpoonup\mu$.

We claim that $\mu$ is invariant: Given $f\in C_b(V)$,
\begin{align*}
\langle f,\P^*\mu \rangle =& \langle \P f,\mu \rangle = \lim_{n\to\infty} \langle \P f,\bar{\mu}_n \rangle= \lim_{n\to\infty} \frac{1}{n}\sum_{k=1}^n\langle \P f,\mu_k \rangle\\
                         =& \lim_{n\to\infty} \frac{1}{n}\sum_{k=1}^n\langle f,\mu_{k+1} \rangle= \lim_{n\to\infty} \frac{1}{n}\left(\sum_{k=1}^n\langle f,\mu_{k}\rangle-\langle f,\mu_0\rangle+\langle f,\mu_{n+1}\rangle\right)\\
                         =& \lim_{n\to\infty} \langle f,\bar{\mu}_n \rangle = \langle f,\mu \rangle,                     
\end{align*}
whence $P^*\mu=\mu$ (Note that for the second inequality we needed the Feller property to ensure that $\P f\in C_b(V)$).

\section{Acknowledgments}  
We thank the second author's advisor, Stamatis Dostoglou, both for suggesting this problem and his helpful conversations and references to the literature.


\appendix
\section{Appendix}

We collect here a few regularity results that were not proved in~\cite{KZ} but which we need. These results were proved in~\cite{Ju} for a different set of boundary conditions but the proofs carry over to our setting.

\subsection{Proof that $v\in C([0,T];V)$}
\label{contintimeSection}

\begin{lem}
\label{stdcont}
Assume that $H$ and $V$ are Hilbert spaces such that the embedding $V \rightarrow H$ is compact.  Let $V'$ be the dual of $V$ and identify $H$ with $H'$ so that $V \subset H \cong H' \subset V'$.  If $v \in L^2(0,T;V)$ and $\partial_t v \in L^2(0,T;V')$, then $v\in C(0,T;H)$ (more precisely, $v = u$ a.e. for some $u \in C([0,T];H)$).
\end{lem}

This is a standard lemma (see Chapter 3, Lemma 1.2 of~\cite{TemamBook}).

\begin{thm}
\label{contintime}
     If $v$ is a strong solution on $[0,T]$, $T>0$, then $v \in C([0,T];V)$.
\end{thm}

\begin{proof}
It suffices to show that 
\begin{equation}
\label{helpful}
     A^{\frac{1}{2}} v \in L^2(0,T;V) \quad \text{ and } \quad 
     \partial_t \big( A^{\frac{1}{2}} v \big) \in L^2(0,T;V'),
\end{equation}
since Lemma \ref{stdcont} will then imply that $A^{\frac{1}{2}} v \in C([0,T];H)$.  Since $\| v \|_V = \| A^{\frac{1}{2}} v \|_H$, this will imply that $v \in C([0,T];V)$.

The first containment is obvious since $v \in L^2(0,T;D(A))$ and $\| v \|_V = \| A^{\frac{1}{2}} v \|_H$.  For the second, let $v_1 \in V$ and multiply (\ref{thePEsAbstract}) by $A^{\frac{1}{2}} v_1$ and integrate to get
\begin{equation*}
     \langle \partial_t v, A^{\frac{1}{2}} v_1 \rangle_H
          =
           -\nu \langle A v, A^{\frac{1}{2}} v_1 \rangle_H
           - \langle B(v,v),A^{\frac{1}{2}} v_1 \rangle_H
           + \langle f, v_1 \rangle_H.
\end{equation*}
By Holder's inequality and the Ladyzenskaya inequalities, we have
\begin{equation*}
\begin{split}
     \langle \partial_t v, A^{\frac{1}{2}} v_1 \rangle_H
          &\leq
           \nu \| A v \|_H \| v_1 \|_V
           +C \| v \|_V^{\frac{3}{2}} \| A v \|_H^{\frac{1}{2}} \| v_1 \|_V \\
           &\quad +C \| v \|_V \| A v \|_H \| v_1 \|_V
           + \| f \|_H \| v_1 \|_V.
\end{split}
\end{equation*}
Write the left hand side as $\langle \partial_t \big( A^{\frac{1}{2}} v \big), v_1 \rangle$, divide by $\| v_1 \|_V$, and take the supremum over all $v_1 \in V$ to get 
\begin{equation}
\label{helpful2}
\begin{split}
     \| \partial_t \big( A^{\frac{1}{2}} v \big) \|_{V'}
     \leq
       \nu \| A v \|_H
       +C \| v \|_V^{\frac{3}{2}} \| A v \|_H^{\frac{1}{2}}
       & +C \| v \|_V \| A v \|_H 
       + \| f \|_H 
\end{split}
\end{equation}
Since the right hand side is square-integrable, we have shown that $\partial_t \big( A^{\frac{1}{2}} v \big) \in L^2(0,T;V')$.
\end{proof}

\subsection{Proof that $v\mapsto S(t)v$ is continuous}
\label{A2}

\begin{thm}
\label{appendixcont}
Given, $v_1$ and $v_2$, two strong solutions to (\ref{thePEsMod}), there exists $C=C(\nu,v_1(0),v_2(0),t)$ such that
\begin{equation*}
     \| v_1(t) - v_2(t) \|_V \leq C \| v_1(0) - v_2(0) \|_V.
\end{equation*}
In particular, $v\in V\mapsto S(t)v\in V$ is continuous for all $t>0$.
\end{thm}
\begin{proof}
$w := v_1 - v_2$ satisfies the equation
\begin{equation*}
     \partial_t w + \nu A w + B(w,v_1) + B(v_2,w) = 0.
\end{equation*}
Multiply by $Aw$ and integrate to get
\begin{equation*}
     \frac{1}{2} \frac{d}{dt} \| w \|_V^2
     +
     \nu \| Aw \|_H^2
     =
     -\langle B(w,v_1),Aw \rangle_H
     -\langle B(v_2,w),Aw \rangle_H.
\end{equation*}
Now estimate the right side using Holder's inequality and the Ladyzenskaya inequalities.
\begin{equation*}
\begin{split}
     \langle B(w,v_1),Aw \rangle_H
     &\leq
     C \| w \|_H^{\frac{1}{4}} \| w \|_V^{\frac{3}{4}} \| v_1 \|_V^{\frac{1}{4}} \| v_1 \|_{V^2}^{\frac{3}{4}} \| Aw \|_H \\
     &\quad +
     C \| w \|_V^{\frac{1}{2}} \| w \|_{V^2}^{\frac{1}{2}} \| v_1 \|_V^{\frac{1}{2}} \| v_1 \|_{V^2}^{\frac{1}{2}} \| Aw \|_H.
\end{split}
\end{equation*}
By Poincare's inequality, Young's inequality, and the fact that $\| v \|_{V^2} \leq C \| Av \|_H$, we have that
\begin{equation*}
\begin{split}
     \langle B(w,v_1),Aw \rangle_H
     &\leq
     \frac{\nu}{2} \| Aw \|_H^2
     +
     \frac{C}{\nu} \| v_1 \|_V^{\frac{1}{2}} \| v_1 \|_{V^2}^{\frac{3}{2}} \| w \|_V^2
     +
     \frac{C}{\nu^3} \| v_1 \|_V^2 \| v_1 \|_{V^2}^2 \| w \|_V^2.
\end{split}
\end{equation*}
By a similar argument, we also have
\begin{equation*}
\begin{split}
     \langle B(v_2,w),Aw \rangle_H
     &\leq
     \frac{\nu}{2} \| Aw \|_H^2
     +
     \frac{C}{\nu^7} \| v_2 \|_V^8 \| w \|_V^2
     +
     \frac{C}{\nu^3} \| v_2 \|_V^2 \| v_2 \|_{V^2}^2 \| w \|_V^2.
\end{split}
\end{equation*}
Therefore,
\begin{equation*}
\begin{split}
     \frac{d}{dt} &\| w \|_V^2 \\
     &\leq
     \Big(
          \frac{C}{\nu} \| v_1 \|_V^{\frac{1}{2}} \| v_1 \|_{V^2}^{\frac{3}{2}}
          +
          \frac{C}{\nu^3} \| v_1 \|_V^2 \| v_1 \|_{V^2}^2
          +
          \frac{C}{\nu^7} \| v_2 \|_V^8
          +
          \frac{C}{\nu^3} \| v_2 \|_V^2 \| v_2 \|_{V^2}^2
     \Big)
     \| w \|_V^2.
\end{split}
\end{equation*}
Since the quantity inside the parentheses is integrable on $[0,t]$, Gronwall's inequality completes the proof.
\end{proof}

\subsection{Proof that $S(t):V\to V$ is a compact operator}
\label{A3}

We first state a lemma which is just a special case of the Aubin-Lions compactness lemma (see Proposition 1.3 of~\cite{ShowalterBook} for that lemma and its proof).

\begin{lem} 
\label{Aubin}
Let
$$
\H := \{v(t): v(t)\in L^2(0,T;V),\ \dot{v}(t)\in L^2(0,T;V')\},
$$
with the norm $\norm{v}_{\H}:=\norm{v}_{L^2(0,T;V)}+\norm{\dot{v}}_{L^2(0,T;V')}$. Then $\H$ is compactly embedded into $L^2(0,T;H)$.
\end{lem}

The following theorem is a central result of~\cite{Ju}. We reproduce the proof here but with slightly different estimates which take into account our different boundary conditions.

\begin{thm}
Let $S(t):V\to V$ be the solution operator to the unforced primitive equations (\ref{thePEsMod}). Then for each $t>0$, $S(t)$ is a compact operator.
\end{thm}

\begin{proof}
Fix $t>0$ and let $\{v_n\}_{n=1}^\infty$ be bounded in $V$. It will suffice to show that $\{S(t)v_n\}_{n=1}^\infty$ has a convergent subsequence in $V$. We have shown the set of paths $\{A^{\frac{1}{2}}S(\cdot)v_n\}_{n=1}^\infty$ is a subset of $\H$ (see (\ref{helpful})). We will now show that this set of paths is in fact bounded in $\H$, i.e.
\begin{equation}
\label{tobeshown}
     \norm{A^{\frac{1}{2}}S(\cdot)v_n}_{L^2(0,T;V)}\leq C \quad \text{ and } \quad 
     \norm{\partial_t\left(A^{\frac{1}{2}}S(\cdot)v_n\right)}_{L^2(0,T;V')}\leq C ,
\end{equation}
where $C\geq0$ is independent of $n$. 

By (\ref{piece2}) we get that $\norm{S(\cdot)v_n}_{L^2(0,T;V)}\leq C$ where $C\geq0$ is independent of $n$. Then, by the same argument used in Section \ref{ConclusionOfProof} we get (c.f. (\ref{boundsupnorm})) that $\sup_{[0,T]}\norm{S(t)v_n}_V\leq C$. Combining these estimates with (\ref{IneqK}) we get that $\int_0^T\norm{\nabla\partial_z S(t)v_n}_{L^2}^2 dt\leq C$. Finally, combining all these bounds with (\ref{Gronwall}) we get the first bound of (\ref{tobeshown}). The second bound of (\ref{tobeshown}) then follows from (\ref{helpful2}).

Boundedness in $\H$ along with Lemma \ref{Aubin} implies that, after passing to a subsequence, $\{A^{\frac{1}{2}}S(\cdot)v_n\}_{n=1}^\infty$ converges in $L^2(0,T;H)$. It follows that $\{S(\cdot)v_n\}_{n=1}^\infty$ converges in $L^2(0,T;V)$ to some limit, which we call $v^*(\cdot)$. It then follows that, after passing to a further subsequence, $S(\tau)v_n\to v^*(\tau)$ in $V$ for almost every $\tau\in[0,T]$.

Since $t>0$, we can find a $\tau^*<t$ such that $S(\tau^*)v_n\to v^*(\tau^*)$. Then by the semigroup property for $S(t)$ and continuity (see Theorem \ref{appendixcont}) we have that
$$
S(t)v_n=S(t-\tau^*)S(\tau^*)v_n \to S(t-\tau^*)v^*(\tau^*),\text{ in }V.
$$
Thus, we have shown that (a subsequence of) $\{S(t)v_n\}_{n=1}^\infty$ converges in $V$.

\end{proof}

\bibliography{Bibliography}{}
\bibliographystyle{plain}

\end{document}